\documentclass[11pt]{article}
\usepackage{times}

\usepackage[margin=1in]{geometry}
\usepackage{graphicx}
\usepackage[pdftex,colorlinks=true,linkcolor=blue,citecolor=blue,urlcolor=black]{hyperref}
\usepackage{amsmath, amsthm, amssymb}
\usepackage{subcaption}
\usepackage{comment}
\usepackage{url}
\usepackage[ruled,lined,linesnumbered]{algorithm2e}
\usepackage{tikz}

\newcommand{\R}{\mathbb{R}}
\newcommand{\N}{\mathbb{N}}

\newcommand{\ket}[1]{| #1 \rangle}

\newcommand{\braket}[1]{\langle #1 \rangle}

\newcommand{\Hil}{\mathcal{H}}

\DeclareMathOperator{\poly}{poly}

\newcommand{\dd}{\mathbf{d}}
\newcommand{\x}{\mathbf{x}}
\newcommand{\y}{\mathbf{y}}

\newcommand{\be}{\begin{equation}}
\newcommand{\ee}{\end{equation}}
\newcommand{\bea}{\begin{eqnarray}}
\newcommand{\eea}{\end{eqnarray}}
\newcommand{\bes}{\begin{equation*}}
\newcommand{\ees}{\end{equation*}}
\newcommand{\beas}{\begin{eqnarray*}}
\newcommand{\eeas}{\end{eqnarray*}}

\makeatletter
\newtheorem*{rep@theorem}{\rep@title}
\newcommand{\newreptheorem}[2]{%
\newenvironment{rep#1}[1]{%
 \def\rep@title{#2 \ref{##1} (restated)}%
 \begin{rep@theorem}}%
 {\end{rep@theorem}}}
\makeatother

\newcommand{\boxalgm}[3]{
\renewcommand{\figurename}{Algorithm}
\begin{figure}[tb]
\begin{center}
\noindent \framebox{
\begin{minipage}{.95\textwidth}
#3
\end{minipage}
}
\caption{#2}
\label{#1}
\end{center}
\end{figure}
\renewcommand{\figurename}{Figure}
}

\newtheorem{thm}{Theorem}
\newtheorem*{thm*}{Theorem}

\newtheorem{lem}[thm]{Lemma}
\newtheorem*{lem*}{Lemma}
\newtheorem{claim}[thm]{Claim}

\newtheorem{dfn}{Definition}
\newreptheorem{thm}{Theorem}
\newreptheorem{lem}{Lemma}
\newreptheorem{cor}{Corollary}

\newcommand{\am}[1]{{\color{red}AM: #1}}

\begin{document}


\title{Quantum speedups of some general-purpose numerical optimisation algorithms}
\author{Cezar-Mihail Alexandru \and Ella Bridgett-Tomkinson \and Noah Linden \and Joseph MacManus \and Ashley Montanaro\thanks{{\tt ashley.montanaro@bristol.ac.uk}} \and Hannah Morris}
\date{\small School of Mathematics, Fry Building, University of Bristol, UK}
\maketitle

\begin{abstract}
We give quantum speedups of several general-purpose numerical optimisation methods for minimising a function $f:\R^n \to \R$. First, we show that many techniques for global optimisation under a Lipschitz constraint can be accelerated near-quadratically. Second, we show that backtracking line search, an ingredient in quasi-Newton optimisation algorithms, can be accelerated up to quadratically. Third, we show that a component of the Nelder-Mead algorithm can be accelerated by up to a multiplicative factor of $O(\sqrt{n})$. Fourth, we show that a quantum gradient computation algorithm of Gily\'en et al.\ can be used to approximately compute gradients in the framework of stochastic gradient descent. In each case, our results are based on applying existing quantum algorithms to accelerate specific components of the classical algorithms, rather than developing new quantum techniques.
\end{abstract}


\section{Introduction}

Quantum computers are designed to use quantum mechanics to outperform their classical counterparts. As well as the remarkable exponential speedups that are known for specialised problems such as integer factorisation and simulation of quantum-mechanical systems, there are also quantum algorithms which speed up general-purpose classical algorithms in the domains of combinatorial search and optimisation. These algorithms may achieve relatively modest speedups, but make up for this by having very broad applications. The most famous example is Grover's algorithm~\cite{grover97}, which achieves a quadratic speedup of classical unstructured search, and can be used to accelerate classical algorithms for solving hard constraint satisfaction problems such as Boolean satisfiability.

Here our focus is on quantum algorithms that accelerate classical numerical optimisation algorithms: that is, algorithms that attempt to solve the problem of finding $\x \in \R^n$ such that $f(\x)$ is minimised, for some function $f:\R^n \rightarrow \R$. (We use boldface throughout for elements of $\R^n$.) A vast number of optimisation algorithms are known. Some algorithms seek to find (or approximate) a global minimum of $f$, given some constraints on $f$; others only attempt to find a local minimum. Some algorithms have provable correctness and/or performance bounds, while the performance of others must be verified experimentally. Whether or not an algorithm has good theoretical properties, its performance on a given problem often can only be determined by running it. These factors have led to the development and use of many numerical optimisation algorithms based on varied techniques.

Here we consider some prominent general-purpose numerical optimisation techniques, and investigate the extent to which they can be accelerated by quantum algorithms. We stress that our goal is not to develop new quantum optimisation techniques (that perhaps would not have rigorous performance bounds), but rather to find quantum algorithms that speed up existing classical techniques, while retaining the same performance guarantees. That is, if the classical algorithm performs well in terms of solution quality or execution time on a given problem instance, the quantum algorithm should also perform well. We assume throughout that the quantum algorithm has access to an oracle that computes $f(\x)$ exactly on particular inputs $\x$, implemented as a quantum circuit\footnote{As we would like to store $\x$ in a register of qubits, technically this is only possible if we consider inputs $\x$ within a bounded region and discretised up to a certain level of precision, and assume that $f(\x)$ is also bounded. However, this is also the case for the corresponding classical algorithms that we accelerate.}. That is, we assume we have access to the map $\ket{\x}\ket{0} \mapsto \ket{\x}\ket{f(\x)}$. This contrasts with a model sometimes used elsewhere in the literature, where $\x$ is assumed to be provided to the quantum algorithm as a quantum state of $\log_2 n$ qubits~\cite{kerenidis17,rebentrost19} stored in a quantum RAM, and the goal is to produce a quantum state corresponding to $\arg\min_{\x} f(\x)$.

Our results can be summarised as follows, where we use the notation (as in the rest of the paper) $T(f)$ for an upper bound on the time required to evaluate the function $f$. See Table \ref{tab:summary} for a summary of the speedups we obtain.

\begin{itemize}
\item Section \ref{sec:lipschitz}: We show that a number of techniques for global optimisation under a Lipschitz constraint can be accelerated near-quadratically, and also discuss some challenges associated with speeding up the related and well-known classical algorithm DIRECT~\cite{jones93}. In Lipschitzian optimisation, one assumes that $|f(\x)-f(\y)| \le K \| \x - \y\|$ for some $K$ that is known in advance (the Lipschitz constant of $f$), where $\|\cdot\|$ is the Euclidean norm. Many techniques for Lipschitzian optimisation can be understood in the framework of branch-and-bound algorithms~\cite{hansen95}. These algorithms are based on dividing $f$'s domain into subsets, and using a lower-bounding procedure to rule out certain subsets from consideration. This enables the use of a quantum algorithm for speeding up branch-and-bound algorithms~\cite{montanaro20}. The complexity of branch-and-bound algorithms is controlled by a parameter $T_{\min}$ discussed below; the quantum algorithm achieves a quadratic reduction in complexity in terms of this parameter. A simple representative example of an algorithm fitting into this framework is Galperin's cubic algorithm~\cite{galperin85}. In this case, the quantum algorithm's complexity is then $\widetilde{O}(\sqrt{T_{\min}} d^{3/2} 2^n T(f) )$, where $d$ is the depth of the branch-and-bound tree, whereas the classical complexity is $O(T_{\min} 2^n T(f))$.

\item Section \ref{sec:quasinewton}: We show that backtracking line search~\cite[Algorithm 3.1]{nocedal06}, a subroutine used in many quasi-Newton optimisation algorithms such as the BFGS algorithm, can be accelerated using a quantum algorithm which is a variant of Grover search~\cite{lin15}. Backtracking line search is based on choosing a direction $\dd$ and searching along that direction. If the overall algorithm makes $k$ iterations, the complexity of choosing $\dd$ is $\tau(\dd)$, and the number of search steps taken by the classical algorithm is $m_0$, the complexity of one iteration of this classical routine is $O(\tau(\dd) + m_0 T(f))$, while the complexity of the quantum algorithm is $O(\tau(\dd) + \sqrt{m_0} (\log k) T(f))$.

\item Section \ref{sec:neldermead}: We show that the Nelder-Mead algorithm~\cite{nelder65}, a widely-used derivative-free numerical optimisation algorithm, can be accelerated using quantum minimum-finding~\cite{durr96}. The algorithm is an iterative procedure based on maintaining a simplex. Assume that $T(f) = \Omega(n^{3/2})$, 
and that the algorithm performs $k$ iterations, $s$ of which are ``shrink'' steps (qv). Then the complexity of the quantum algorithm is $O(((s+1)\sqrt{n} \log k + k)T(f)))$, as compared with the classical complexity, $O(((s+1) n + k) T(f))$. So if the number of shrink steps is large with respect to $k$, or $k$ is small, the quantum speedup can be relatively substantial (up to a $O(\sqrt{n})$ factor).

\item Section \ref{sec:sgd}: Approximate computation of a gradient is a key subroutine in many optimisation algorithms, including the very widely-used gradient descent algorithm~\cite{bottou18}. We show that the gradient of functions $f$ of the form $f(\x) = \frac{1}{N} \sum_{i=1}^N f_i(\x)$ can be computed more efficiently using a quantum algorithm of Gily\'en, Arunachalam and Wiebe~\cite{gilyen19a}. Given that each individual function $f_i$ is bounded and can be computed in time $T(f)$ (and satisfies some technical constraints on its partial derivatives), the quantum algorithm outputs an approximation of the gradient that is accurate up to $\epsilon$ in the $\ell_\infty$ norm, in time $\widetilde{O}(\sqrt{n} T(f) \epsilon^{-1})$, as compared with the classical complexity $\widetilde{O}(n T(f) \epsilon^{-2})$. (The $\widetilde{O}$ notation hides polylogarithmic factors in $N$, $n$ and $1/\epsilon$.) However, as we will discuss, it is not clear whether this notion of approximation is sufficient to accelerate classical stochastic gradient descent algorithms.
\end{itemize}

\begin{table}
\begin{center}
\resizebox{\textwidth}{!}{%
\begin{tabular}{|c|c|c|c|c|}
\hline \bf \S & \bf Algorithm & \bf Classical & \bf Quantum & \bf Technique\\
\hline \ref{sec:lipschitz} & Global opt.\ w/Lipschitz constraint & (e.g.) $O(T_{\min} 2^n T(f))$ & $\widetilde{O}(\sqrt{T_{\min}} 2^n T(f) d^{3/2})$ & Branch-and-bound~\cite{montanaro20} \\
\ref{sec:quasinewton} & Backtracking line search & $O(k(\tau(\dd) + m_{\max} T(f)))$ & $O(k(\tau(\dd) + \sqrt{m_{\max}} (\log k) T(f)))$ & Variant of Grover's algorithm~\cite{lin15} \\
\ref{sec:neldermead} & Nelder-Mead & $O(((s+1) n + k) T(f))$ & $O(((s+1)\sqrt{n} \log k + k)T(f)))$ & Quantum minimum-finding~\cite{durr96} \\
\ref{sec:sgd} & Gradients of averaged functions & $\widetilde{O}(n T(f) \epsilon^{-2})$ & $\widetilde{O}(\sqrt{n} T(f) \epsilon^{-1})$ & Quantum gradient computation~\cite{gilyen19a} \\
\hline
\end{tabular}
}
\end{center}
\caption{Informal summary of the results obtained in this paper. Parameters for algorithms are described in the respective sections of the paper, and summarised as follows. $T(f)$: complexity of computing $f:\R^n \to \R$; $T_{\min}$: size of a truncated branch-and-bound tree; $d$: depth of a branch-and-bound tree; $\tau(\dd)$: complexity of computing a descent direction; $k$: number of iterations; $m_{\max}$: worst-case number of backtracking line search steps; $s$: number of simplex shrinking steps; $\epsilon$: accuracy. The bounds make various assumptions about $f$ that are detailed in the text.}
\label{tab:summary}
\end{table}

In each case, the quantum speedups we find are based on the use of existing quantum algorithms, rather than the development of new algorithmic techniques. We believe that there are many more quantum speedups of numerical optimisation algorithms to be discovered. We remark that, in many of the cases we consider, the extent of the quantum speedup achieved depends on the interplay of various parameters governing the optimisation algorithm's runtime, so not every problem instance will yield a speedup.

Prior work on quantum speedups of numerical optimisation algorithms (as opposed to the analysis of {\em new} quantum algorithms such as the adiabatic algorithm~\cite{farhi00} or quantum approximate optimisation algorithm~\cite{hogg00,farhi14}) has been relatively limited. D\"urr and H\o yer~\cite{durr96} gave a quantum algorithm to find a global minimum of a function $f$ on a discrete space of size $N$, which is based on the use of Grover's algorithm and uses $O(\sqrt{N})$ evaluations of $f$. Arunachalam~\cite{arunachalam14} applied D\"urr and H\o yer's algorithm to improve the generalised pattern search and mesh-adaptive direct search optimisation algorithms. A sequence of papers has found quantum speedups of linear programming and semidefinite programming algorithms~\cite{brandao17,apeldoorn17,apeldoorn18,kerenidis18,brandao19}; quantum speedups of more general convex optimisation algorithms are also known~\cite{apeldoorn18a,chakrabarti18}. Quantum speedups are known for computing gradients~\cite{jordan05,gilyen19a,cornelissen19}, an important subroutine in many optimisation algorithms; larger (exponential) speedups could be available in gradient descent-type algorithms if the inputs to the optimisation algorithm are available in a quantum RAM (qRAM)~\cite{kerenidis17,rebentrost19}. Recently, it was shown that classical algorithms based on the general technique known as branch-and-bound can be accelerated near-quadratically~\cite{montanaro20}.


\section{Branch-and-bound algorithms for global optimisation with a Lipschitz constraint}
\label{sec:lipschitz}

Finding a global minimum of an arbitrary function $f:\R^n \rightarrow \R$ can be a very challenging (or indeed impossible) task. One way to make this problem more tractable is to assume that $f$ satisfies a Lipschitz condition: $|f(\x)-f(\y)| \le K \| \x - \y\|$ for some $K$ that is known in advance, where $\|\cdot\|$ is the Euclidean norm. Finding a global minimum of $f$ under this condition is known as Lipschitzian optimisation. Lipschitzian optimisation is very general and hence can be applied in many contexts. Hansen and Jaumard~\cite{hansen95} describe a selection of applications of Lipschitzian optimisation, including solution of nonlinear equations and inequalities; parametrisation of statistical models; black box system optimisation; and location problems.

It is natural to restrict the domain of $f$ to $[0,1]^n$, and to assume that $f$ is bounded such that $f(\x) \in [0,1]$ for all $\x \in [0,1]^n$. Finally, we can relax to solving the approximate optimisation problem of finding $\mathbf{y}$ such that $f(\y) - \min_{\mathbf{x} \in [0,1]^n} f(\x) \le \epsilon$, for some accuracy parameter $\epsilon$ that is determined in advance. Even in the case $n=1$ and with these restrictions, this problem is far from trivial. One class of algorithms that can solve Lipschitzian optimisation problems are branch-and-bound algorithms. Generically, a branch-and-bound algorithm solves a minimisation problem using the following procedures:
\begin{itemize}
\item A branching procedure which, given a subset $S$ of possible solutions, divides $S$ into two or more smaller subsets, or returns that $S$ should not be divided further.
\item A bounding procedure which, when given a subset $S$ produced during the branching process, returns a lower bound $L(S)$ such that $L(S) \le \min_{\x \in S} f(\x)$.
\end{itemize}
Branch-and-bound algorithms can be seen as exploring a tree, whose vertices correspond to subsets $S$. The children of a subset $S$ correspond to the subsets which $S$ was divided into, and leaves are subsets that should not be divided further. For a leaf, one should additionally have that $L(S) = \min_{\x \in S} f(\x)$. Branch-and-bound algorithms use the additional information provided by the branch and bound procedures to explore the most promising sets $S$ early on, and to avoid exploring subsets $S$ such that $L(S)$ is larger than the best solution found so far. One can show that the complexity of an optimal classical branch-and-bound algorithm based on these generic procedures is controlled by the size of the branch-and-bound tree, truncated by deleting all vertices whose corresponding lower bounds are less than the optimal cost $\min_{\x} f(\x)$: if the size of this tree is $T_{\min}$, the optimal classical algorithm makes $\Theta(T_{\min})$ calls to the branch and bound procedures~\cite{karp93}. It is not required to know $T_{\min}$ in order to apply this bound.

A generic framework for branch-and-bound algorithms in the context of Lipschitzian optimisation was given by Hansen and Jaumard~\cite[Section 3.3]{hansen95}, and we describe it as Algorithm \ref{alg:genbab}. The algorithm splits $[0,1]^n$ into hyperrectangles $I$, each of which is recursively split again. Each hyperrectangle has an associated upper bound (obtained by evaluating $f$ at a discrete set of points in that hyperrectangle) and lower bound (obtained via a separate lower-bounding function), and the algorithm terminates when it finds a hyperrectangle whose upper bound is sufficiently close to its lower bound. Convergence is guaranteed if some simple criteria are satisfied, discussed in~\cite{hansen95} (for example, the upper bound and lower bound should converge as the interval size tends to 0). Hansen and Jaumard show that many previously known algorithms for Lipschitzian optimisation can be understood as particular cases of Algorithm \ref{alg:genbab}. These include Galperin's cubic algorithm~\cite{galperin85}, which proceeds by dividing the search space into hypercubes, and algorithms of Pijavskii~\cite{pijavskii72}, Shubert~\cite{shubert72} and Mladineo~\cite{mladineo85}.

\boxalgm{alg:genbab}{Generic branch-and-bound algorithm for Lipschitzian optimisation problems~\cite{hansen95}}{
\begin{enumerate}
\item Choose a discrete set $D \subset [0,1]^n$ and set $f_{opt} \leftarrow \min_{\x \in D} f(\x)$; $\mathbf{x}_{opt} \leftarrow \arg\min_{\x \in D} \x$\\ {\small \color{blue}[Initialise upper bound]}
\item Let $F$ be a lower-bounding function of $f$ on $[0,1]^n$ and compute $F_{opt} = \min_{\x \in [0,1]^n} F(\x)$ {\small \color{blue}[Initialise lower bound]}
\item If $f_{opt} - F_{opt} \le \epsilon$, stop. Otherwise, $\mathcal{L} \leftarrow \mathcal{P} = ([0,1]^n,F_{opt})$ {\small \color{blue}[Initialise branch-and-bound tree]}
\item While $\mathcal{L}$ is nonempty:
\begin{enumerate}
\item Let $\mathcal{L}'$ be a subset of $\mathcal{L}$ chosen according to a selection rule
\item For each subproblem $\mathcal{P} = (I,F_{opt})$ in $\mathcal{L}'$:
\begin{enumerate}
\item Partition $I$ into hyperrectangles $I_1,\dots,I_p$ according to a branching rule {\small \color{blue}[Branch]}
\item For $j=1,\dots,p$:
\begin{enumerate}
\item Choose a discrete set $D_j \subset I_j$. For all $\x \in D_j$:
\begin{itemize}
\item If $f(\x) < f_{opt}$ then $f_{opt} \leftarrow f(\x)$, and $\x_{opt} \leftarrow \x$ {\small \color{blue}[Update upper bound]}
\end{itemize}
\item Compute $F^j_{opt} = \min_{\x \in I_j} F^j(\x)$, where $F^j$ is a lower-bounding function on $I_j$ {\small \color{blue}[Compute lower bound]}
\item If $f^j_{opt} - F_{opt} \le \epsilon$:
\begin{itemize}
\item then if $\mathbf{x}_{opt} \in D_j$, $f_{opt}$ is an $\epsilon$-optimal solution of problem $(I_j,F^j_{opt})$
\item else add $\mathcal{P}_j = (I_j,F^j_{opt})$ to $\mathcal{L}$ {\small \color{blue}[Explore interval $I_j$ further]}
\end{itemize}
\item Delete from $\mathcal{L}$ all subproblems $\mathcal{P}$ with $F_{opt} \ge f_{opt}$.
\end{enumerate}
\end{enumerate}
\end{enumerate}
\end{enumerate}
}

The branching procedure of Algorithm \ref{alg:genbab} fits into the standard branch-and-bound framework. Given a subset $I_j$, an upper bound is obtained by evaluating $f(\x)$ at a discrete set of positions $\x$, and a lower bound is obtained using the bounding function $F^j$. If the two are within $\epsilon$, $I_j$ should not be expanded further. Otherwise, $I_j$ is split into subsets. Algorithm \ref{alg:genbab} has a notion of selecting the next subset in $\mathcal{L}$ using a selection rule, but it is shown in~\cite{karp93} that the best possible selection rule in branch-and-bound procedures (in a query complexity sense) is to expand the subset whose bounding function is smallest\footnote{The proof of this is based on the intuition that the algorithm cannot rule out subsets whose lower bound is smaller than the cost of the optimal solution. In the setting of Lipschitzian optimisation, this only holds if the lower bounding rule is tight, in the sense that given a lower bound on $f(\x)$, for $\x \in I_j$, there exists a Lipschitz function $f$ such that this lower bound is achieved.}.

There is a quantum algorithm that can achieve a near-quadratic speedup of classical branch-and-bound algorithms~\cite{montanaro20}. The algorithm is based on the use of quantum procedures for estimating the size of an unknown tree~\cite{ambainis17}, and searching within such a tree~\cite{belovs13,belovs13a,montanaro18}. The algorithm achieves a complexity of $\widetilde{O}(\sqrt{T_{\min}} d^{3/2})$ uses of the branch and bound procedures for finding the minimum of $f$ up to accuracy $\epsilon$. In this bound $d$ is the maximal depth of the branch-and-bound tree and the $\widetilde{O}$ notation hides polylogarithmic factors in $d$, $1/\epsilon$, and $1/\delta$, where $\delta$ is the probability of failure. (We remark that the algorithm as presented in~\cite{montanaro20} assumes knowledge of an upper bound on $d$ in advance, but such a bound can be found efficiently by applying the quantum tree search algorithms of~\cite{montanaro18,belovs13,belovs13a} to the branch-and-bound tree obtained by truncating at depth $d'$, with exponentially increasing choices of $d'$, until $d'$ is found where the corresponding tree does not contain any internal vertices that have not been expanded.)

The quantum branch-and-bound algorithm can immediately be applied to Algorithm \ref{alg:genbab}. If the time complexity of the branching and bounding rules is upper-bounded by $C$, the cost of the quantum algorithm is $\widetilde{O}(\sqrt{T_{\min}} d^{3/2} C)$, as compared with the classical complexity, which is $O(T_{\min} C)$. If $T_{\min} \gg d$, the speedup of the quantum algorithm over its classical counterpart in terms of the number of uses of the branching and bounding rules is near-quadratic. If these rules in turn are relatively simple to compute compared with $T_{\min}$ (as is likely to be the case for challenging optimisation problems that occur in practice), this translates into a near-quadratic runtime speedup.

To illustrate how this approach could be applied in practice, a simple example of an algorithm fitting into this framework is Galperin's cubic algorithm~\cite{galperin85}. The branch and bound procedures are defined as follows, recalling that $K$ is the Lipschitz constant of $f$:
\begin{itemize}
\item Branch: the subproblem $I$ corresponding to a hypercube is divided into $p=q^n$ equal hypercubes, for some $q \ge 2$, by dividing each side into $q$ equal parts.
\item Lower bounding rule: Let $\mathbf{x_0}$ be an extreme point of $I$. $I$ has side length $1/q^k$ for some integer $k$. Then a lower bound is $f(\mathbf{x_0}) - \frac{K}{q^k} \sqrt{n}$, maximised over extreme points of $I$.
\item Upper bounding rule: Evaluate $f$ on the extreme points of $I$ and return the minimum value found.
\end{itemize}

Galperin's algorithm is illustrated in Figure \ref{fig:galperin} for the case $n=1$. The complexity of the branch and bounding steps is dominated by the cost of evaluating $f$ at the extreme points of each hypercube $I$, which is $O(2^n T(f))$. The quantum complexity is then $\widetilde{O}(\sqrt{T_{\min}} d^{3/2} 2^n T(f) )$, whereas the classical complexity is $O(T_{\min} 2^n T(f))$; so we see that the speedup is largest for small $n$, e.g.\ $n=O(1)$.

\begin{figure}[t]
\begin{center}
\begin{tikzpicture}[xscale=15]
\draw[yshift=1cm,smooth,blue,domain=0:1] plot ({\x},{3*\x*\x-\x});
\draw (0,0) -- (1,0);
\draw (0,-0.1) -- (0,0.1) node[above] {0};
\draw (0.125,-0.1)  node[below,red] {\begin{tabular}{c} \\ \color{blue}-1 \end{tabular}} -- (0.125,0.1) node[above] {};
\draw (0.25,-0.1)  node[below,red] {\begin{tabular}{c} 2 \\ \color{blue}-2 \end{tabular}} -- (0.25,0.1) node[above] {-0.31};
\draw (0.3125,-0.1)  node[below,red] {\begin{tabular}{c}  \\ \color{blue}-0.81 \end{tabular}} -- (0.3125,0.1) node[above] {};
\draw (0.375,-0.1)  node[below,red] {\begin{tabular}{c} 3 \\ \color{blue}-1.25 \end{tabular}} -- (0.375,0.1) node[above] {-0.33};
\draw (0.4375,-0.1)  node[below,red] {\begin{tabular}{c}  \\ \color{blue}-0.75 \end{tabular}} -- (0.4375,0.1) node[above] {};
\draw (0.5,-0.1) node[below,red] {\begin{tabular}{c} 1 \\ \color{blue}-3 \end{tabular}} -- (0.5,0.1) node[above] {-0.25};
\draw (0.75,-0.1)  node[below,red] {\begin{tabular}{c} \\ \color{blue}-1 \end{tabular}} -- (0.75,0.1) node[above] {};
\draw (1,-0.1) -- (1,0.1) node[above] {1};
\end{tikzpicture}
\end{center}
\caption{Galperin's cubic algorithm for $n=1$, $q=2$ applied to the function $f(x) = 3x^2 - 2x$ (plotted in blue) with Lipschitz bound $K=4$, which is minimised at $x=1/3$ with $f(x)=-1/3$. The result of a few steps of splitting into subintervals is shown. The centres of intervals are labelled below with the step at which they are divided into subintervals (red), and the lower bound in that interval (blue). Endpoints are labelled above with the evaluated function values, shown to two decimal places.}
\label{fig:galperin}
\end{figure}
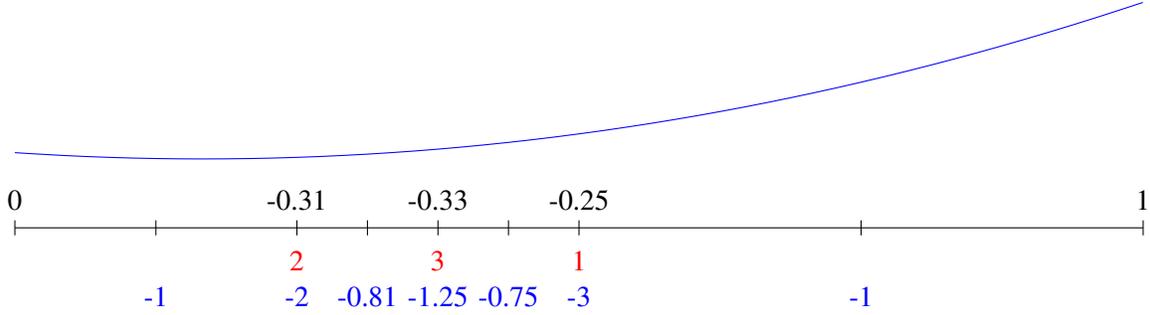

\subsection{The DIRECT algorithm}

A prominent algorithm proposed to handle Lipschitzian optimisation for $n$-variate functions where one does not know the Lipschitz constant in advance is known as DIRECT~\cite{jones93} (for ``dividing rectangles''). The basic concept is to divide $[0,1]^n$ into (hyper)rectangles, and at each step of the algorithm to produce a list of potentially optimal rectangles, which are those that should be expanded further; see Appendix \ref{app:direct} for more details. This is similar to the branch-and-bound algorithms of the previous section, but with the additional complication of generating the list of potentially optimal rectangles, which involves interaction across several nodes of the branch-and-bound tree. This creates a difficulty for the quantum branch-and-bound algorithm, as it can only use branch and bound procedures based on only local information from the tree. Therefore it is unclear whether a similar quadratic speedup can be obtained.

To identify the potentially optimal vertices, the DIRECT algorithm uses a 2d convex hull algorithm. It is a natural idea to speed this up via a quantum convex hull algorithm. Lanzagorta and Uhlmann~\cite{lanzagorta10} have described a quantum algorithm based on Grover's algorithm for computing a convex hull of $m$ points in 2d with complexity $O(\sqrt{m}h)$, where $h$ is the number of points in the convex hull; they also give an algorithm based on a heuristic whose runtime may be $O(\sqrt{mh})$ for practically relevant problems. However, the special case of the convex hull problem that is relevant to DIRECT can be solved in time $O(h)$~\cite{jones93}, so this does not lead to an overall quantum speedup.


\section{Backtracking line search}
\label{sec:quasinewton}

Backtracking line search\footnote{Not to be confused with the combinatorial optimisation technique known as backtracking.}~\cite{nocedal06} is a line search optimisation algorithm devised by Armijo in 1966~\cite{armijo66}. The goal of a line search method is, given a starting point $\mathbf{x_0} \in \R^n$ and a direction $\mathbf{d}$, to move to a new point $\mathbf{x_0} + \eta \mathbf{d}$ in the direction $\mathbf{d}$, in order to minimise a function $f:\R^n \rightarrow \R$. Backtracking line search is a particular line search technique based on the use of an exponentially decreasing parameter $\eta$.
A generic optimisation method based on backtracking line search is described as Algorithm \ref{alg:bls}. In this section we describe a quantum speedup of this algorithm.

\boxalgm{alg:bls}{Generic line search method based on backtracking line search}{
\begin{enumerate}
\item  Choose a starting point $\mathbf{x_0}$ and constants $\gamma \in (0,1)$ and $\beta \in (0,1)$. Set $\x \gets \mathbf{x_0}$.
\item Choose a direction $\mathbf{d}$ such that $D_\mathbf{d}f(x) < 0$, where $D_\mathbf{d}f = \mathbf{d} \cdot \nabla f$ is the directional derivative in direction $\mathbf{d}$. If no such $\mathbf{d}$ exists ($\nabla f = 0$), terminate.
\item Compute the step size: $\eta = \gamma^{m_0}$, where $m_0 = \min \{ m \in \mathbb{N} \mid f(\mathbf{x} + \gamma^m \mathbf{d}) \le f(\x) + \beta \gamma^m D_\mathbf{d}(f) \}$. As $D_\mathbf{d}f(\x) < 0$, $m_0$ always exists.
\item Set $\x \gets \x + \gamma^{m_0} \mathbf{d}$.
\item Go back to step 2 if termination condition is not met (number of iterations, threshold, etc.)
\end{enumerate}
}

Different approaches can be used to choose $\mathbf{d}$. These include:
\begin{itemize}
\item Steepest descent: $\mathbf{d} \propto -\nabla f(\x)$.
\item Newton's method: $\mathbf{d} \propto -H(\x)^{-1} \nabla f(\x)$, where $H(\x)$ is the Hessian of $f$.
\item Quasi-Newton methods (such as BFGS): $\mathbf{d} \propto -B(\x)^{-1} \nabla f(\x)$, where $B(\x)$ is some approximation of $H(\x)$.
\end{itemize}
Let $\tau(\dd)$ denote the complexity of choosing the direction $\mathbf{d}$; note that $\tau(\dd) = \Omega(n)$, because just writing down $\mathbf{d}$ requires time $\Omega(n)$. Then the overall complexity of one iteration of Algorithm \ref{alg:bls} is $O(\tau(\dd) + m_0 T(f))$.
We can reduce this complexity using the following result of Lin and Lin~\cite{lin15} (see also~\cite{kothari14}):

\begin{thm}[Lin and Lin~\cite{lin15}]
\label{thm:firstelem}
Consider a function $g: \{1,\dots,N\} \rightarrow \{0,1\}$. Let $m = \min \{ y: g(y) = 1\}$, if this set is nonempty, or otherwise $m=\infty$. Then there is a quantum algorithm that succeeds with probability at least 0.99 and outputs $m$ using $O(\sqrt{m})$ evaluations of $g$ if $m\neq \infty$, and otherwise outputs that $m=\infty$ in $O(\sqrt{N})$ steps.
\end{thm}

We apply this result to step 3 of the classical algorithm to achieve a square-root reduction in the dependence on $m_0$. To achieve a final probability of failure bounded by a small constant, by a union bound over the $k$ iterations, it is sufficient to repeat the algorithm of Theorem \ref{thm:firstelem} $O(\log k)$ times to achieve $O(1/k)$ failure probability at each iteration. This gives an overall complexity of the quantum algorithm which is $O(\tau(\dd) + \sqrt{m_0}(\log k) T(f))$ per iteration. If the overall algorithm makes $k$ iterations, and $m_{\max}$ is the largest value of $m_0$ for any iteration, we have an overall complexity of $O(k(\tau(\dd) + \sqrt{m_{\max}}(\log k) T(f)))$. In cases where $\tau(\dd) = O(n)$ (such as the steepest descent method), $T(f) = \Omega(n)$, and $k$ is not exponentially large in $n$, the dominant term in this complexity bound is the second one, and we always achieve a quantum speedup. The assumption $T(f) = \Omega(n)$ is natural if $f$ depends on all $n$ variables. 

This condition $f(\mathbf{x} + \eta \mathbf{d}) \le f(\x) + \beta \eta D_\mathbf{d}(f)$ that is used in step 3 is called the Armijo condition. If $\nabla f$ is Lipschitz at $\x$ with Lipschitz constant $L(\x)$ ($\|\nabla f(\x) - \nabla f(\mathbf{y})\| \le L(\x) \|\x - \mathbf{y}\|$), any
  \[ \eta \in \left[0, \frac{2(\gamma - 1) D_\mathbf{d}f(\x)}{L(\x)\|\dd\|^2}\right] \]
  satisfies the Armijo condition~\cite[Theorem 2.1]{gould2003introduction}. If we choose $\mathbf{d}$ such that $\|\mathbf{d}\|=1$, then since $\eta = \gamma^{m_0}$, $\gamma = \Omega(1)$, $m_0 = O(\log(L(\x) / |D_\mathbf{d}f(\x)|))$. Therefore, the speedup achieved by the quantum algorithm (based on this worst-case bound) will be greatest when $L$ is large (representing that $\nabla f$ could change rapidly), yet $|D_\mathbf{d}f(\x)|$ is small (representing that $f$ does not change rapidly in direction $\dd$).
  
Another way in which one might hope to speed up Algorithm \ref{alg:bls} is computing $D_\mathbf{d}f(\x)$ more efficiently. For example, a quantum algorithm was presented by Gily\'en, Arunachalam and Wiebe~\cite{gilyen19}, based on a detailed analysis of and modifications to an earlier algorithm of Jordan~\cite{jordan05}, that approximately computes $\nabla f(\x)$ for smooth functions quadratically more efficiently than classical methods (that are based e.g.\ on finite differences). However, it seems challenging to prove that such an approximation can be inserted in the backtracking line search framework without affecting the performance of the overall algorithm, in the worst case. This is because even a small change in the direction $\mathbf{d}$ can significantly change the behaviour of the algorithm, as the definition of Step 3 of Algorithm \ref{alg:bls} is such that an arbitrarily small change to the values taken by $f$ along the direction $\mathbf{d}$ can change $m_0$ substantially. See Section \ref{sec:sgd} below for a further discussion of this algorithm.

Finally, we remark that one simple way to find a direction $\mathbf{d}$ such that $D_{\mathbf{d}}(f)$ is nonzero, as required for the line search procedure, is to choose $i$ such that $\partial f / \partial \x_i$ is nonzero. Although a valid choice, in practice this could be less efficient than (for example) moving in the direction of steepest descent. The use of Grover's algorithm would reduce the complexity of this step to $O(\sqrt{n}(\log k) T(f))$, as compared with the classical $O(n T(f))$.
  

\section{Nelder-Mead algorithm}
\label{sec:neldermead}

The Nelder-Mead algorithm is a direct search optimisation algorithm; that is, one which does not require information about the gradient of the objective function. It is commonly-used and implemented within many computer algebra packages. However, little convergence theory exists and in practice it is ineffective in higher dimensions\footnote{Indeed, according to Lagarias et al.~\cite{lagarias98}, ``given all the known inefficiencies and failures of the Nelder-Mead algorithm\dots\ one might wonder why it is used at all, let alone why it is so extraordinarily popular.''.}~\cite{lagarias98,han06}. The Nelder-Mead algorithm uses expansion, reflection, contraction and shrink steps to update a simplex in $\R^n$.
A number of variants of the algorithm have been proposed. The variant we will use was analysed by Lagarias et al.~\cite{lagarias98}, and is presented as Algorithm \ref{alg:neldermead}. Algorithm \ref{alg:neldermead} does not specify a termination criterion. Termination criteria that could be used include the function values at the simplex points becoming sufficiently close; the simplex points themselves becoming sufficiently close; or an iteration limit being reached.

\boxalgm{alg:neldermead}{Nelder-Mead algorithm (see e.g.~\cite{lagarias98})}{
Let $\alpha$, $\beta$, $\gamma$, $\delta$ be parameters defined such that $\alpha>0$, $\beta>1$, $\beta>\alpha$, $0<\gamma<1$, $0<\delta<1$. Standard choices are $\alpha=1$, $\beta=2$, $\gamma=\frac{1}{2}$, $\delta=\frac{1}{2}$.
\begin{enumerate}
\item \textbf{Initialise.} Define an $n$-dimensional simplex $S$ with $n+1$ vertices, $S=\{\mathbf{x_0},\dots,\mathbf{x_n}\}$.
\item \textbf{Sort.} Order and relabel the vertices of the simplex such that $f(\mathbf{x_0})\geq f(\mathbf{x_1}) \geq\dots\geq f(\mathbf{x_n})$ and let $\mathbf{x_0}$ be the worst vertex, $\mathbf{x_1}$ the next-worst vertex and $\mathbf{x_n}$ the best vertex. Set $\mathbf{c}= \frac{1}{n} \sum^{n}_{i=1}\mathbf{x_i}$.
\item \textbf{Reflection.} Calculate the reflection point, $\mathbf{x_r}=\mathbf{c}+\alpha(\mathbf{c}-\mathbf{x_0})$. If $f(\mathbf{x_n})\leq f(\mathbf{x_r}) < f(\mathbf{x_1})$ accept reflection, replace $\mathbf{x_0}$ with $\mathbf{x_r}$ and return to step $2$.
\item \textbf{Expansion.} If $f(\mathbf{x_r}) < f(\mathbf{x_n})$, calculate the expansion point $\mathbf{x_e}=\mathbf{c}+\beta(\mathbf{x_r}-\mathbf{c})$. If $f(\mathbf{x_e})<f(\mathbf{x_r})$ accept the expansion point and replace $\mathbf{x_0}$ with $\mathbf{x_e}$, otherwise accept the reflection point and replace $\mathbf{x_0}$ with $\mathbf{x_r}$. Return to step $2$.
\item \textbf{Outside Contraction.} If $f(\mathbf{x_1})\leq f(\mathbf{x_r})<f(\mathbf{x_0})$, compute the outer contraction point, $\mathbf{x_{c1}}=\mathbf{c}+\gamma(\mathbf{x_r}-\mathbf{c})$. If $f(\mathbf{x_{c1}})\leq f(\mathbf{x_r})$, accept the outside contraction point, replace $\mathbf{x_0}$ with $\mathbf{x_{c1}}$ and return to step $2$. Else go to step $7$.
\item \textbf{Inside Contraction.} If $f(\mathbf{x_r})\geq f(\mathbf{x_0})$, calculate the inside contraction point, $\mathbf{x_{c2}}=\mathbf{c}-\gamma(\mathbf{c}-\mathbf{x_0})$. If $f(\mathbf{x_{c2}})<f(\mathbf{x_0})$ accept inside contraction, replace $\mathbf{x_0}$ with $\mathbf{x_{c2}}$ and return to step $2$. Else go to step $7$.
\item \textbf{Shrink.} For all points other than the best point, replace it with its shrink point, $\mathbf{x_i} = \delta \mathbf{x_i} + (1-\delta) \mathbf{x_n}$ for $i=0,...,n-1$. Go to step 2.
\end{enumerate}
}

In this section we describe a quantum speedup of the Nelder-Mead algorithm. We first determine the classical complexity of the algorithm, drawing on the analysis of~\cite{singer04}. The complexity of step 1 is $O(n^2)$ to write down the $n+1$ points. To analyse step 2, observe that a complete ordering of the points is never required; the only information about the ordering needed is the worst vertex $\mathbf{x_0}$, the next-worst vertex $\mathbf{x_1}$, and the best vertex $\mathbf{x_n}$. Knowledge of the identities of these points is sufficient to compute the centroid $\mathbf{c}$, and to carry out all the updates required, including the shrink step. So the first time that step 2 is executed, its complexity is $O(n^2 + n T(f))$, where the $O(n^2)$ comes from computing the centroid. Each time step 2 is executed subsequently, except following a shrink step, the required updates can be made in time $O(n)$. The complexity of each of steps 3 to 6 is $O(n + T(f))$; step 7 is $O(n^2)$. So the complexity of performing $k$ iterations, of which $s$ include a shrink step, is $O((s+1)(n^2 + n T(f)) + k(n+T(f)))$. If $T(f) = \Omega(n)$, this simplifies to $O(((s+1) n + k) T(f))$. 

The complexity of step 2, when executed for the first time or following a shrink step, can be improved using quantum minimum-finding:
\begin{thm}[D\"urr and H\o yer~\cite{durr96}]
\label{thm:quantumminimum}
Given a function $h:[N] \to \R$ and $\epsilon > 0$, there is a quantum algorithm that outputs $\arg\min_x h(x)$ with probability at least $1-\epsilon$ using $O(\sqrt{N} \log 1/\epsilon)$ evaluations of $h$.
\end{thm}
Thus a quantum algorithm using Theorem \ref{thm:quantumminimum} can find the worst, next-worst and best vertices with failure probability $O(1/k)$ at each iteration in time $O(\sqrt{n} T(f) \log k)$ in total. This choice of failure probability is so that, by a union bound, the total probability of failure can be bounded by an arbitrarily small constant. Further, observe that the centroid can be updated in time $O(n)$ following a shrink step, as if $\mathbf{c'}$ denotes the updated centroid, then $\mathbf{c'} = \delta \mathbf{c} + (1-\delta) \mathbf{x_n}$. This does not give a quantum speedup of step 2 in all cases; the first time that step 2 is executed, if $T(f) = O(n^{3/2})$, its complexity is dominated by the $O(n^2)$ cost of computing the centroid. There also remains an $O(n^2)$ cost for updating the points at each shrink step. (There may be a more efficient way of keeping track of these shrink steps; however, we do not pursue this further here.)
Then the overall complexity of the quantum algorithm is $O((s+1)(n^2+\sqrt{n} T(f) \log k) + k(n + T(f)))$, and using a union bound over the $k$ steps, the algorithm's failure probability is bounded above by an arbitrarily small constant. If $T(f) = \Omega(n^{3/2})$, this simplifies to $O(((s+1)\sqrt{n} \log k + k)T(f))$. Comparing with the classical complexity, we see that the quantum speedup is largest when $s$ is large compared with $k$.

However, in practice shrink steps appear to be rare; in one set of experiments, only 33 shrink steps were observed in 2.9M iterations~\cite{torczon89}, and shrink steps never occur when Nelder-Mead is applied to a strictly convex function~\cite{lagarias98}. If there are no shrink steps and $T(f) = \Omega(n^{3/2})$, the complexity of the quantum algorithm is $O((\sqrt{n} \log k + k)T(f))$, while the complexity of the classical algorithm is $O((n+k)T(f))$. This is still a quantum speedup if $k = o(n)$; on the other hand, if $k = \Omega(n)$, the complexity is dominated by evaluating $f$ once at each iteration, and it is difficult to see how a quantum speedup could be achieved.

To be able to use quantum minimum-finding, we have assumed the ability to construct superpositions of the form $\frac{1}{\sqrt{n+1}} \sum_{i=0}^n \ket{i}\ket{\mathbf{x_i}}$, which enables us to evaluate $f$ in superposition. This is a quantum RAM~\cite{giovannetti08}, and quantum RAMs are often assumed to be difficult to construct; however, our requirements are very weak, because we only need the addressing to be performed in time $\widetilde{O}(n)$, rather than $O(\log n)$, which can be achieved using an explicit quantum circuit.

Finally, we consider the possibility of accelerating calculation of the centroid $\mathbf{c}$ using a quantum algorithm.
If each component of each vector $\mathbf{x_i}$ is suitably bounded (e.g.~$\|\mathbf{x_i}\|_\infty \le 1$) we could use quantum mean estimation~\cite{heinrich01,brassard11,montanaro15} to estimate each component of $\mathbf{c}$ up to accuracy $\epsilon$ in time $O((n/\epsilon) \log (n/\epsilon))$ with failure probability bounded by a small constant, where the $\log (n/\epsilon)$ term comes from reducing the failure probability for each component to $\epsilon / n$. Classical mean estimation could be used instead with an overhead of an additional $O(1/\epsilon)$ factor. This would give an overall time complexity similar to that derived above, but it is not obvious what the effect of replacing the centroid with an approximate centroid would be on the overall algorithm. For example, it is argued in~\cite{fajfar18} that random perturbations to the centroid throughout the algorithm can be beneficial.


\section{Stochastic gradient descent}
\label{sec:sgd}

One of the most widely-used, effective and simple methods for finding a local minimum of a function is gradient descent. Given a function $f:\R^n \to \R$ and an initial point $\x \in \R^n$, the algorithm moves to the point $\mathbf{x'} = \x - \eta \nabla f(\x)$, where $\eta > 0$. In application areas such as machine learning~\cite{bottou18}, one often encounters functions $f$ of the form
\be \label{eq:stochasticf} f(\x) = \frac{1}{N} \sum_{i=1}^N f_i(\x) \ee
for some ``simple'' functions $f_i(\x)$, where $N$ is large. (For example, $f_i(\x)$ could be the error of a neural network parametrised by $\x$ on the $i$'th item of training data, and we might seek to minimise the average error.) Rather than computing the exact gradient $\nabla f(\x)$ by summing $\nabla f_i(\x)$ over all $N$ choices for $i$, it is natural to approximate $\nabla f$ by sampling $k$ random indices $i_1,\dots,i_k \in [N]$ with replacement and outputting $\frac{1}{k}(\nabla f_{i_1}(\x) + \dots + \nabla f_{i_k}(\x))$. (The case $k=1$ is known as stochastic gradient descent; the sample $i_1,\dots,i_k$ is sometimes known as a mini-batch.) If $f$ satisfies the Lipschitz condition that $\|\nabla f_i(\x)\|_\infty \le 1$, to approximate $\nabla f(\x)$ up to additive error $\epsilon$ in the $\ell_\infty$ norm with failure probability $\delta$ it is sufficient to take $k = O(\epsilon^{-2} \log(n/\delta))$ by a Chernoff bound argument. Let $T(f)$ denote an upper bound on the time required to compute $f_i(\x)$ for all $i$. If we approximate $\nabla f_i(\x)$ using the finite difference method, then each approximation to $\nabla f_i(\x)$ can be computed in time $O(n T(f))$, giving a total complexity of $O(n T(f) \epsilon^{-2} \log(n/\delta))$.

The use of quantum amplitude estimation~\cite{brassard02} would improve the dependence on $\epsilon$ quadratically. Here we observe that the dependence on $n$ can also be improved quadratically, using a result of Gily\'en, Arunachalam and Wiebe~\cite{gilyen19a}. We will impose the restriction (for technical reasons) that the range of each function $f_i$ is within $[1/10,9/10]$, where these numbers could be replaced with any constants between 0 and 1. Given the more typical constraint that $f_i:\R^n \to [0,1]$ (e.g.\ if the output of $f_i$ represents a probability), $f_i$ can easily be modified to satisfy this constraint by a simple linear transformation, which does not change $\arg\min_{\x} f(\x)$.

The results of~\cite{gilyen19a} use two somewhat nonstandard oracle models which we now define. First we will consider probability access, and define what a \textit{probability oracle} is.

\begin{dfn}[Probability oracle]
Let $g : Z \to [0,1]$, where $\{ \ket z : z \in Z\}$ forms an orthonormal basis of the Hilbert space $\Hil$, and let $\Hil_A$ be an ancilla register on $q > 0$ qubits. Then an operator $U_g : \Hil \otimes \Hil_A \to \Hil \otimes \Hil_A$ is called a probability oracle for $g$ if 
$$
U_g \ket z \ket{0^n} = \ket z \left( \sqrt {1-g(z)} \ket {\psi_0} \ket 0 + \sqrt {g(z)} \ket{\psi_1} \ket 1 \right)
$$
for some arbitrary $q-1$ qubit states $\ket {\psi_0}$, $\ket {\psi_1}$. 
\end{dfn}

Essentially, within this model our objective function corresponds precisely to the probability of a certain outcome being observed upon measurement (in particular, the probability of seeing $\ket 1$ when measuring the final qubit). Indeed, given a classical description of the function $g(z)$, an oracle of this form can be constructed without a significant overhead~\cite{cao13}. The next access model we consider is access via a \textit{phase oracle}.

\begin{dfn}[Phase oracle]
Given a function $g : Z \to [0,1]$, and given that $\{ \ket z : z \in Z\}$ forms an orthonormal basis of the Hilbert space $\Hil$, then the corresponding \textit{phase oracle} $O_g : \Hil \to \Hil$ allows queries of the form
\[ O_g \ket z = e^{ig(z)} \ket z. \]
\end{dfn}

The authors of \cite{gilyen19a} showed that a probability oracle is capable of simulating a phase oracle, and vice versa, with only logarithmic overhead:

\begin{thm}[Converting between probability and phase oracles~\cite{gilyen19a}] \label{thm:probtophase}
Suppose $g : Z \to [0,1]$ is given by access to a probability oracle $U_g$ which makes use of $a$ auxiliary qubits. Then we can simulate an $\epsilon$-approximate phase oracle using $O (\log(1/\epsilon))$ queries to $U_g$; the gate complexity is the same up to a factor of $O(a)$. Similarly, suppose $g:Z \to [\delta, 1-\delta]$ is given by access to a phase oracle $O_g$. Then we can construct an $\epsilon$-approximate probability oracle for $g$ using $O (\log (1/\epsilon)/\delta)$ queries to $O_g$. The gate complexity is the same up to a factor of $O(\log(1/\epsilon) (\log \log(1/\epsilon)+ \log(1/\delta)))$.
\end{thm}

What this shows is that the two access models are more-or-less equivalent in power.
%
Now we have defined probability oracles, we can show that access to probability oracles for the individual $f_i$ functions immediately gives such access for $f$ itself.

\begin{lem}\label{lem:losstoobj}
Assume we have access to each function $f_i:\R^n \to [0,1]$ via a probability oracle $U_{f_i}$. Then we can construct a probability oracle for $f$ with a single use of controlled-$U_{f_i}$ operations (in superposition) and $O(\log N)$ additional operations.
\end{lem}

\begin{proof}
We start with the superposition $\frac{1}{\sqrt{N}} \sum_{i=1}^N \ket{i} \ket{\x} \ket{0}$, where $\ket{\x}$ denotes a description of the real vector $\x$ in terms of binary, up to some digits of precision, leading to an orthonormal basis. If $N$ is a power of 2, this state can be constructed easily by applying Hadamard gates to each qubit in a register of $\log_2 N$ qubits. If not, the state can be constructed in circuit complexity $O(\log N)$ as follows: attach a register of $\lceil \log_2 N \rceil$ qubits; apply Hadamard gates to produce $\frac{1}{2^{\lceil \log_2 N \rceil}} \sum_{i=1}^{2^{\lceil \log_2 N \rceil}} \ket{i}$; compute the function ``$i \le N$'' into an ancilla qubit using an efficient comparison circuit (e.g.~\cite{draper06}); measure the ancilla qubit; and proceed only if the answer is 1. If not (which occurs with probability at most $1/2$), repeat this step. We then apply the controlled operation $\ket{i}\ket{\psi} \mapsto \ket{i} U_{f_i} \ket{\psi}$. This produces
\[ \frac{1}{\sqrt{N}} \sum_{i=1}^N \ket{i} \ket{\x} \left( \sqrt{1-f_i(\x)} \ket{\psi^{(i)}_0} \ket0 + \sqrt{f_i(\x)} \ket{\psi^{(i)}_1} \ket1 \right) \]
for some sequences of normalised states $\ket{\psi^{(i)}_0}$, $\ket{\psi^{(i)}_1}$. Rearranging subsystems, we can write this as
\[ \ket{\x}\ket{\psi_0}\ket{0} + \ket{\x} \ket{\psi_1}\ket{1} \]
for some unnormalised states $\ket{\psi_0}$, $\ket{\psi_1}$ where $\braket{\psi_1|\psi_1} = \frac{1}{N} \sum_{i=1}^N f_i(\x) = f(\x)$ as required by the definition of a probability oracle for $f$.
\end{proof}

We will use this probability oracle within the framework of the fast quantum algorithm of~\cite{gilyen19a} for computing gradients. This algorithm is applicable to functions that satisfy a certain smoothness condition. Given some analytic function $h : \R^n \to \R$, let $\partial_i h(\x) = \frac {\partial}{\partial \x_i} h(\x)$, and for any $k \in \mathbb N$, $\alpha = (\alpha_1, \ldots, \alpha_k) \in [n]^k$, let 
\[ \partial_\alpha h(\x) = \partial_{\alpha_1} \ldots \partial_{\alpha_k} h(\x). \]
The following result shows that if each function $f_i$ satisfies the required smoothness condition~\cite{gilyen19a}, we have that the overall function $f$ also satisfies the same condition.

\begin{claim}\label{thm:losssmooth}
Let $c$ be a real constant, and fix some $\x \in \R^n$. Suppose that for all $i \in [N]$ the function $f_i : \R^n \to \R$ is analytic, and that for every natural number $k$, and $\alpha \in [n]^k$, we have that 
\[ | \partial_\alpha f_i (\x) | \leq c^k k^{\frac k 2}, \]
then we have that $f$ also satisfies the same condition. 
\end{claim}

\begin{proof}
We apply the linearity of $\partial_\alpha$. Observe that 
\[ |\partial_\alpha f(\x) | = \left| \frac 1 N \sum_i \partial_\alpha f_i (\x) \right| \leq \frac 1 N \sum_i |\partial_\alpha f_i (\x) | \leq c^k k^{\frac k 2}, \]
and we are done.
\end{proof}

In fact it's not too hard to see that this claim generalises to more-or-less any bound on the partial derivatives. We can now state the result we will need from~\cite{gilyen19a}.

\begin{thm}[Gily\'en, Arunachalam and Wiebe~{\cite[Theorem 25]{gilyen19a}}]
\label{thm:fastgradient}
Suppose that $g : \R^n \to \R$ is an analytic function such that, for all $r \in \N$ and $\alpha \in [n]^r$, $|\partial_\alpha g(\x)| \le c^r r^{r/2}$. Assume access to $g$ is given by a phase oracle $O_g$. Then there exists an algorithm that outputs a vector $\widetilde{\nabla f}(\x) \in \R^n$ such that  
$
\Vert\widetilde{\nabla f}(\x) - \nabla f (\x) \Vert _\infty \leq \epsilon
$
with 99\% probability, using $\widetilde{O}(\sqrt n /\epsilon)$ queries to the oracle and additional time $\widetilde{O}(n^{3/2} / \epsilon)$.
\end{thm}

Note that, if the time complexity of evaluating $O_g$ is $\Omega(n)$, this dominates the overall runtime bound. We can encapsulate the combination of these results in the following theorem.

\begin{thm}
\label{thm:gradient}
Let $f$ be defined as in (\ref{eq:stochasticf}), and assume that each function $f_i$ satisfies the conditions required for Theorem \ref{thm:fastgradient} and can be computed in time $T(f)$, for some bound $T(f)$ such that $T(f) = \Omega(n)$. Then there is a quantum algorithm that outputs $\widetilde{\nabla f}(\x)$ such that $\| \widetilde{\nabla f}(\x) - \nabla f(\x) \|_\infty \le \epsilon$ with 99\% probability, in time $\widetilde{O}(\sqrt{n} T(f) \epsilon^{-1})$.
\end{thm}

\begin{proof}
Given the ability to compute each $f_i$ function in time $T(f)$, we can produce a phase oracle computing $f_i$ in time $O(T(f))$. By Theorem \ref{thm:probtophase}, and using that $f_i:\R^n \rightarrow [1/10,9/10]$, we can then obtain an operation approximating a probability oracle for $f_i$ up to error $\epsilon$ in time $\widetilde{O}(T(f))$. By Lemma \ref{lem:losstoobj}, this gives a probability oracle for $f$, at additional cost $O(\log N)$. By Theorem \ref{thm:probtophase}, we then obtain a phase oracle for $f$ at additional cost $\poly\log(N,1/\epsilon)$. This finally allows us to apply Theorem \ref{thm:fastgradient} to achieve the stated complexity.
\end{proof}

Despite Theorem \ref{thm:gradient} giving a more efficient quantum algorithm for approximately computing $\nabla f$, it is not clear whether this translates into a more efficient quantum algorithm for stochastic gradient descent, or a quantum speedup of other algorithms making use of $\nabla f$. This is because the algorithm of~\cite{gilyen19a} only outputs an approximate gradient, and one which may not be an unbiased estimate of $\nabla f$. To prove approximate convergence of stochastic gradient descent, it is not essential for the gradient estimates to be unbiased~\cite{bottou18}, and it is plausible that an approximate estimate of the gradient should lead to an approximate minimiser for $f$ being found. However, the technique used in~\cite{bottou18} to show approximate convergence in this scenario requires the 2-norm of the approximate gradient to be close to that of $\nabla f$. The algorithm of~\cite{gilyen19a} provides accuracy $\epsilon$ in the $\infty$-norm, which would only give accuracy $\epsilon\sqrt{n}$ in the 2-norm. Further, it was shown by Cornelissen~\cite{cornelissen19} that if $f$ is picked from a certain class of smooth functions, approximating $\nabla f$ up to 2-norm accuracy $\epsilon$ requires $\Omega(n / \epsilon)$ uses of a phase oracle for $f$ in the worst case, so this is not merely a technical restriction. Nevertheless, it is possible that quantum gradient estimation may be more efficient than stochastic gradient descent in practice.


\subsection*{Acknowledgements}

We would like to thank Srinivasan Arunachalam for helpful explanations of the results of~\cite{gilyen19a}. We acknowledge support from the QuantERA ERA-NET Cofund in Quantum Technologies implemented within the European Union's Horizon 2020 Programme (QuantAlgo project) and EPSRC grants EP/R043957/1 and EP/T001062/1. This project has received funding from the European Research Council (ERC) under the European Union's Horizon 2020 research and innovation programme (grant agreement No.\ 817581). No new data were created during this study.


\appendix

\section{The DIRECT algorithm}
\label{app:direct}

In this appendix we briefly describe the DIRECT (``dividing rectangles'') algorithm~\cite{jones93} for global optimisation of functions $f:[0,1]^n \to \R$, which is presented as Algorithm \ref{alg:direct}. The algorithm is based on maintaining a partition of the hypercube into hyperrectangles using the concept of ``potentially optimal'' hyperrectangles: 

\begin{dfn}
\label{dfn:potopt}
Let $\epsilon>0$, let $f_{\min}$ be the current best function value found, and let $m$ be the current number of hyperrectangles in the partition of $[0,1]^n$. Let $\mathbf{c_i}$ denote the centre of the $i$th hyperrectangle, and let $d_i$ denote the distance from the centre to the vertices. Hyperrectangle $j$ is said to be potentially optimal if there exists $\widetilde{K}>0$ such that $\forall i=1,\dots,m$,
\begin{equation}
\label{eq:1}
    f(\mathbf{c_j})-\widetilde{K}d_j\leq f(\mathbf{c_i})-\widetilde{K}d_i
\end{equation}
and 
\begin{equation}
\label{eq:2}
    f(\mathbf{c_j})-\widetilde{K}d_j\leq f_{\min}-\epsilon|f_{\min}|.
\end{equation}
\end{dfn}

We think of $\widetilde{K}$ in Definition \ref{dfn:potopt} as a surrogate for the Lipschitz constant of $f$ (which is not assumed to be known in advance). An example of the first couple of steps of dividing $[0,1]^2$ into rectangles is shown in Figure \ref{fig:Dividing the Initial Hypercube}. The set of potentially optimal hyperrectangles can be determined in time $O(m')$, where $m' \le m$ is the number of distinct interval lengths, using a convex hull technique described in~\cite{jones93} and illustrated in Figure \ref{fig:Potentially Optimal Hyperrectangles}. The conditions (\ref{eq:1}) and (\ref{eq:2}) are satisfied by the points that lie on the lower convex hull when $f(\mathbf{c}_j)$ is plotted against $d_j$ for each hyperrectangle, and we also include the point $(0,f_{\min}-\epsilon|f_{\min}|)$. In Figure \ref{fig:Potentially Optimal Hyperrectangles} the red dots represent potentially optimal hyperrectangles whereas the black dots represent hyperrectangles that are not potentially optimal.

\boxalgm{alg:direct}{DIRECT algorithm~\cite{jones93} for optimisation over $[0,1]^n$.}{
\begin{enumerate}
    \item Let $\mathbf{c}$ be the centre of $[0,1]^n$ and evaluate $f(\mathbf{c})$. Assign $f_{\min} \leftarrow f(\mathbf{c})$,  $m \leftarrow 1$,  $t \leftarrow 1$.
    \item Let $S$ be the set of potentially optimal hyperrectangles.
    \item Select any hyperrectangle $j \in S$.
    \item Evaluate hyperrectangle $j$ and decide where to divide it using the following procedure:
    \begin{enumerate}
    \item Let $I$ be the set of dimensions with maximal side length. Let $\delta$ be one-third of this maximal side length. Let $\mathbf{c}$ be the centre of hyperrectangle $j$.
    \item Evaluate $f$ at the points $\mathbf{c} \pm \delta \mathbf{e_i}$ for all $i \in I$, where $\mathbf{e_i}$ is the $i$'th vector in the standard basis.
    \item Divide the hyperrectangle containing $\mathbf{c}$ into thirds along the dimensions $i \in I$, in ascending order of $w_i = \min \{f(\mathbf{c} + \delta \mathbf{e_i}), f(\mathbf{c} + \delta \mathbf{e_i}) \}$.
    \end{enumerate}
    Let $\Delta m$ be the number of new points evaluated. Update $m \leftarrow m+\Delta m$, $f_{\min} \leftarrow $ new best min.
    \item $S \leftarrow S-\{j\}$. If $S \neq \emptyset$, go to step 3.
    \item $t \leftarrow t+1$. If $t=T$, where $T$ is the iteration limit, then stop, if not go to step 2.
    \end{enumerate}
}

\begin{figure}
\centering
\begin{subfigure}[t]{0.49\textwidth}
\centering
\begin{tikzpicture}
\draw (0,0) -- (6,0); 
\draw (0,0) -- (0,6);
\draw (0,6) -- (6,6);
\draw (6,0) -- (6,6);
\draw (2,0) -- (2,6);
\draw (4,0) -- (4,6);
\draw (2,2) -- (4,2);
\draw (2,4) -- (4,4);
\draw[fill=black] (3,1) circle (0.08);
\draw[fill=black] (3,5) circle (0.08);
\draw[fill=red!90!yellow!80!, draw=red!90!yellow!80!] (3,3) circle (0.08);
\draw[fill=black] (1,3) circle (0.08);
\draw[fill=black] (5,3) circle (0.08);
\node[below] at (3,1) {3};
\node[below] at (3,5) {2};
\node[below] at (1,3) {4};
\node[below] at (5,3) {5};
\end{tikzpicture}
\caption{Dividing the initial hypercube $[0,1]^n$} \label{fig:Dividing the Initial Hypercube}
\end{subfigure}
\begin{subfigure}[t]{0.49\textwidth}
\centering
\begin{tikzpicture}
\draw[->] (0,0) -- (5,0) node[right] {$d_j$}; 
\draw[->] (0,0) -- (0,5) node[above] {$f(\mathbf{c_j})$};
\draw (0,0.7) -- (2,1.5);
\draw (2,1.5) -- (3.25,2.5);
\draw (3.25,2.5) -- (4,3.5);
\draw[dotted] (0,1.4) -- (1.2,1.4);
\node[left] at (0,1.4) {$f_{\min}$};
\node[left] at (0,0.7) {$f_{\min}-\epsilon|f_{\min}|$};
\draw[fill=red!90!yellow!80!, draw=red!90!yellow!80!] (2,1.5) circle (0.08);
\draw[fill=red!90!yellow!80!, draw=red!90!yellow!80!] (3.25,2.5) circle (0.08);
\draw[fill=red!90!yellow!80!, draw=red!90!yellow!80!] (4,3.5) circle (0.08);
\draw[fill=cyan!40!green!70!blue!80!, draw=cyan!40!green!70!blue!80!] (0,0.7) circle (0.08);
\draw[fill=black] (1.2,1.4) circle (0.08);
\draw[fill=black] (1.2,2) circle (0.08);
\draw[fill=black] (1.2,2.4) circle (0.08);
\draw[fill=black] (1.2,3.5) circle (0.08);
\draw[fill=black] (1.2,4.3) circle (0.08);
\draw[fill=black] (2,2.8) circle (0.08);
\draw[fill=black] (2,4) circle (0.08);
\draw[fill=black] (3.25,3.35) circle (0.08);
\draw[fill=black] (3.25,3.75) circle (0.08);
\draw[fill=black] (4,4) circle (0.08);
\end{tikzpicture}
\caption{Identifying potentially optimal hyperrectangles} \label{fig:Potentially Optimal Hyperrectangles}
\end{subfigure}
\caption{Illustration of aspects of the DIRECT algorithm}
\end{figure}
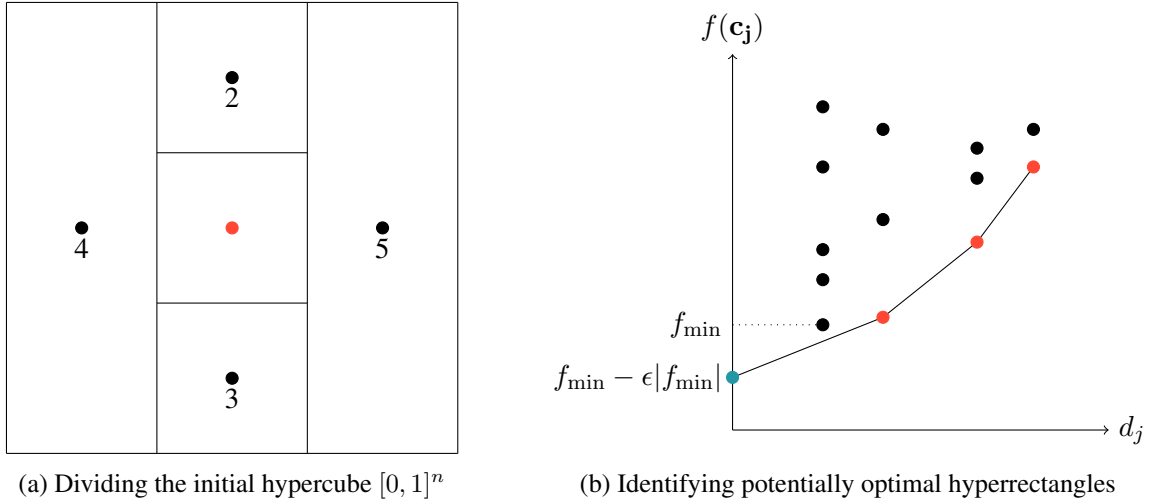

\bibliographystyle{plain}
\bibliography{../../thesis}

\end{document}